\documentclass[pra,twocolumn,superscriptaddress,nobibnotes]{revtex4-1}
\usepackage{graphicx}
\usepackage{amssymb,amsmath,color,amsfonts, amsthm}
\usepackage{hyperref}
\usepackage{ragged2e}
\usepackage[capitalise]{cleveref}

\hypersetup{
	colorlinks,
	linkcolor={blue},
	citecolor={blue},
	urlcolor={blue}
}

\theoremstyle{plain}
\newtheorem{thm}{Theorem}
\newtheorem{lem}[thm]{Lemma}

\DeclareMathOperator{\tr}{Tr}

\newcommand{\ket}[1]{|#1\rangle}
\newcommand{\tn}[1]{^{\otimes #1}}
\newcommand{\mc}{\ensuremath{\mathcal}}

\newcommand{\ideal}[1]{\ensuremath{\mc{#1}}}
\newcommand{\noisy}[1]{\ensuremath{\tilde{\ideal{#1}}}}

\begin{document}

\title{Characterizing large-scale quantum computers via cycle benchmarking}

\date{\today}

\author{Alexander Erhard}
\thanks{These authors contributed equally to this work. Contact: alexander.erhard@uibk.ac.at, jwallman@uwaterloo.ca}
\affiliation{Institute for Experimental Physics, University of Innsbruck, 6020 Innsbruck, Austria}

\author{Joel J. Wallman}
\thanks{These authors contributed equally to this work. Contact: alexander.erhard@uibk.ac.at, jwallman@uwaterloo.ca}
\affiliation{Institute for Quantum Computing and Department of Applied Mathematics, University of Waterloo, Waterloo, Canada}
\affiliation{Quantum Benchmark Inc., Kitchener, ON N2H 4C3, Canada}

\author{Lukas Postler}
\affiliation{Institute for Experimental Physics, University of Innsbruck, 6020 Innsbruck, Austria}

\author{Michael Meth}
\affiliation{Institute for Experimental Physics, University of Innsbruck, 6020 Innsbruck, Austria}

\author{Roman Stricker}
\affiliation{Institute for Experimental Physics, University of Innsbruck, 6020 Innsbruck, Austria}

\author{Esteban A. Martinez}
\affiliation{Institute for Experimental Physics, University of Innsbruck, 6020 Innsbruck, Austria}
\affiliation{Niels Bohr Institute, University of Copenhagen, 2100 Copenhagen, Denmark}

\author{Philipp Schindler}
\affiliation{Institute for Experimental Physics, University of Innsbruck, 6020 Innsbruck, Austria}

\author{Thomas Monz}
\affiliation{Institute for Experimental Physics, University of Innsbruck, 6020 Innsbruck, Austria}

\author{Joseph Emerson}
\affiliation{Institute for Quantum Computing and Department of Applied Mathematics, University of Waterloo, Waterloo, Canada}
\affiliation{Quantum Benchmark Inc., Kitchener, ON N2H 4C3, Canada}

\author{Rainer Blatt}
\affiliation{Institute for Experimental Physics, University of Innsbruck, 6020 Innsbruck, Austria}
\affiliation{Institute for Quantum Optics and Quantum Information of the Austrian Academy of Sciences, 6020 Innsbruck, Austria}

\begin{abstract}
Quantum computers promise to solve certain problems more efficiently than their digital counterparts.
A major challenge towards practically useful quantum computing is characterizing and reducing the various errors that accumulate during an algorithm running on large-scale processors.
Current characterization techniques are unable to adequately account for the exponentially large set of potential errors, including cross-talk and other correlated noise sources.
Here we develop cycle benchmarking, a rigorous and practically scalable protocol for characterizing local and global errors across multi-qubit quantum processors.
We experimentally demonstrate its practicality by quantifying such errors in non-entangling and entangling operations on an ion-trap quantum computer with up to 10 qubits, with total process fidelities for multi-qubit entangling gates ranging from $99.6(1)$\,\% for 2 qubits to $86(2)$\,\% for 10 qubits.
Furthermore, cycle benchmarking data validates that the error rate per single-qubit gate and per two-qubit coupling does not increase with increasing system size.
\end{abstract}

\maketitle

Practical methods to characterize quantum processes acting on large-scale quantum systems are required to assess current devices and steer the development of future, more powerful, devices. In principle, quantum processes can be fully characterized using, for example, quantum process tomography~\cite{Chuang1997} or gate set tomography~\cite{Merkel2012,Blume-Kohout2013,Blume-Kohout2017}.
However, any protocol for fully characterizing a quantum process requires a number of experiments and digital post-processing resources that grows exponentially with the number of qubits, even with improvements such as compressed sensing~\cite{Flammia2012,Rodionov2014}.
As a result, the largest quantum processes that have been fully characterized to date acted only on three qubits~\cite{Weinstein2004}.

The exponential resources required for a full characterization can be circumvented by extracting partial information about quantum processes.
A partial characterization typically yields some figure of merit, such as the process fidelity~\footnote{The process fidelity, also known as the process fidelity, is equivalent to the more commonly used average gate fidelity up to a dimensional factor that is approximately 1~\cite{Horodecki1999,Nielsen2002}.
However, the process fidelity is more convenient because it satisfies $F(\noisy{G}\otimes\noisy{H}, \ideal{G}\otimes\ideal{H}) = F(\noisy{G}, \ideal{G})F(\noisy{H}, \ideal{H})$ for any two quantum processes \ideal{G} and \ideal{H}, as can readily be verified from \cref{eq:processFidelity}.}, comparing the noisy implementation of a quantum process to the desired operation. 

The process fidelity can be efficiently estimated by randomized benchmarking~\cite{Emerson2005,Dankert2009,Magesan2011} or direct fidelity estimation~\cite{Flammia2011,DaSilva2011,Moussa2012}. 
Direct fidelity estimation can be efficient and hence has been implemented for up to 7 qubits~\cite{lu2015experimental}, but conflates state preparation and measurement (SPAM) errors with the process fidelity, limiting its value for realistic systems.
SPAM errors increase with the system size and so robustness to SPAM is increasingly important for many qubits.
Randomized benchmarking decouples the SPAM errors from gate operation errors by applying multiple random elements of the $N$-qubit Clifford group ~\cite{Dankert2009,Magesan2011}. 
However, implementing each Clifford operation requires $\mathcal{O}(N^2/\log{N})$ primitive two-qubit operations~\cite{Aaronson2004}, so that randomized benchmarking provides very coarse information about the primitive operations.
Furthermore, for error rates as low as $0.1\%$ per two-qubit  operation, a single 10-qubit Clifford operation will have a cumulative error rate on the order of $10\%$, which substantially increases the number of measurements required to accurately estimate the process fidelity.
Due to these practical limitations, randomized benchmarking has only been applied on operations involving three or less qubits~\cite{McKay2017}. 
While randomized benchmarking can be performed on small subsets of the qubit register \cite{Gambetta2012}, such experiments do not explore the full Hilbert space and therefore will not detect important performance-limiting error mechanisms such as cross-talk.
Most crucially, undetected cross-talk and other spatially correlated errors will typically require much higher overheads in fault-tolerant quantum error correction schemes~\cite{Preskill2013}.
Hence characterizing all significant errors affecting an entire register is a critical prerequisite for scalable quantum computation.
To achieve this, we focus on the concept of a \emph{cycle} of operations (introduced in \cite{Wallman2016a}), which is a set of operations that act on an entire quantum register within a set period of time, in analogy to a digital clock cycle.

In this paper, we introduce cycle benchmarking (CB), a protocol for estimating the effect of all global and local error mechanisms that occur when a clock cycle of operations is applied to a quantum register.
We prove that CB is robust to SPAM errors and that the number of measurements required to estimate the process fidelity to a fixed precision is approximately independent of the number of qubits.
We demonstrate the practicality of CB for many-qubit systems by using it to experimentally estimate the process fidelity of both non-entangling Pauli operations and the multi-qubit entangling M{\o}lmer-S{\o}rensen (MS) gate~\cite{sorensen1999quantum, sorensen2000entanglement} acting on up to ten qubits.
We also confirm that the protocol and analysis methods, derived under theoretical assumptions, produce consistent results in our experimental system.

Mathematically, the ideal operation of interest is described by the corresponding unitary matrix $G$. Its action is expressed by a map $\ideal{G}: \rho\to G \rho G^\dagger$ that acts on the state of the quantum register, described by the density matrix $\rho$. We denote the map of an ideal operation by capital calligraphic letters, such as $\ideal{G}$, and their noisy experimental implementations will be indicated by an overset tilde, such as $\noisy{G}$.
We denote the composition of gates by the natural matrix operations for the map representation, so, e.g., \ideal{R}\ideal{G} means first apply \ideal{G} then apply \ideal{R}, and $\ideal{G}^m$ means apply \ideal{G} a total of $m$ times. A particularly important class of processes are Pauli cycles $\ideal{P}$, where the unitary matrix of the process is the $N$-qubit Pauli matrix $P$.

We evaluate the quality of a noisy process \noisy{G} by its process fidelity to the ideal target \ideal{G}, which can be written as~\cite{Flammia2011}
\begin{align}\label{eq:processFidelity}
    F(\noisy{G},\ideal{G}) = \sum_{P\in \{I,X,Y,Z\}^{\otimes N}} 4^{-N}F_P(\noisy{G},\ideal{G}),
\end{align}
where
\begin{align}\label{eq:PauliFidelity}
F_P(\noisy{G},\ideal{G}) = 2^{-N} \tr\left[\ideal{G}(P)\noisy{G}(P)\right].
\end{align}

Each quantity $F_P(\noisy{G},\ideal{G})$ can be experimentally estimated by preparing an eigenstate of $P$, applying the noisy gate $\noisy{G}$, and then measuring the expectation value of the ideal outcome $\ideal{G}(P)$. 
The process fidelity may be estimated by averaging $F_P(\noisy{G},\ideal{G})$ over a set of Pauli matrices.
However, a sampling protocol (as in direct fidelity estimation~\cite{Flammia2011,DaSilva2011}) for estimating these individual terms is not robust to SPAM errors.
Robustness to SPAM is particularly important because SPAM errors can dominate the gate errors.

We now outline how the CB protocol can quantify the effect of global and local error mechanisms affecting different primitive cycle operations of interest.
Inspired by randomized benchmarking~\cite{Emerson2005}, SPAM errors can be decoupled from the process fidelity by applying the noisy operation of interest $\noisy{G}$ a total of $m$ times and extracting the process fidelity from the decay of $F_P(\noisy{G}^m,\ideal{G}^m)$ as a function of the sequence length $m$.
Extracting a meaningful error per application of the gate of interest is nontrivial for generic noise channels~\cite{Carignan-Dugas2015}.
However, the process fidelity can be rigorously extracted from the decay of $F_P(\noisy{G}^m,\ideal{G}^m)$ with $m$ if the noise process is a Pauli channel.
We can engineer an effective Pauli noise channel by introducing a round of random Pauli cycles at each time step between each application of the cycle of interest~\cite{Knill2005} and the overhead for this randomization can then be eliminated via randomized compiling~\cite{Wallman2016a}.
The effective noise is then associated with the composition of $\ideal{G}$ with a random Pauli cycle, called a dressed cycle, which is an important characterization primitive for any algorithm implemented via randomized compiling~\cite{Wallman2016a}.
Therefore the estimated quantity is the average of the process fidelities of the composite cycle of $\noisy{G}$ combined with a uniformly random Pauli cycle \noisy{R},
\begin{align}\label{eq:compositeFidelity}
F_{\mathrm{RC}}(\noisy{G}, \ideal{G}) =    \sum_{\ideal{R}\in\{\ideal{I},\ideal{X},\ideal{Y},\ideal{Z}\}^{\otimes N}} 4^{-N} F(\noisy{G}\noisy{R},\ideal{G}\ideal{R}).
\end{align}
The process fidelity of the noise on $\noisy{G}$ alone may also be estimated by taking the ratio of the estimates obtained for \noisy{G} and the identity process \noisy{I}, in analogy to interleaved benchmarking~\cite{Magesan2012}.
It should be noted that this method of estimating the fidelity of the noise on \noisy{G} alone is generally subject to a large systematic uncertainty~\cite{Carignan-Dugas}, so the CB method is most precise in the important context of characterizing errors on dressed cycles~\cite{Wallman2016a}.

\begin{figure*}
    \centering
    \includegraphics[width=0.9\textwidth]{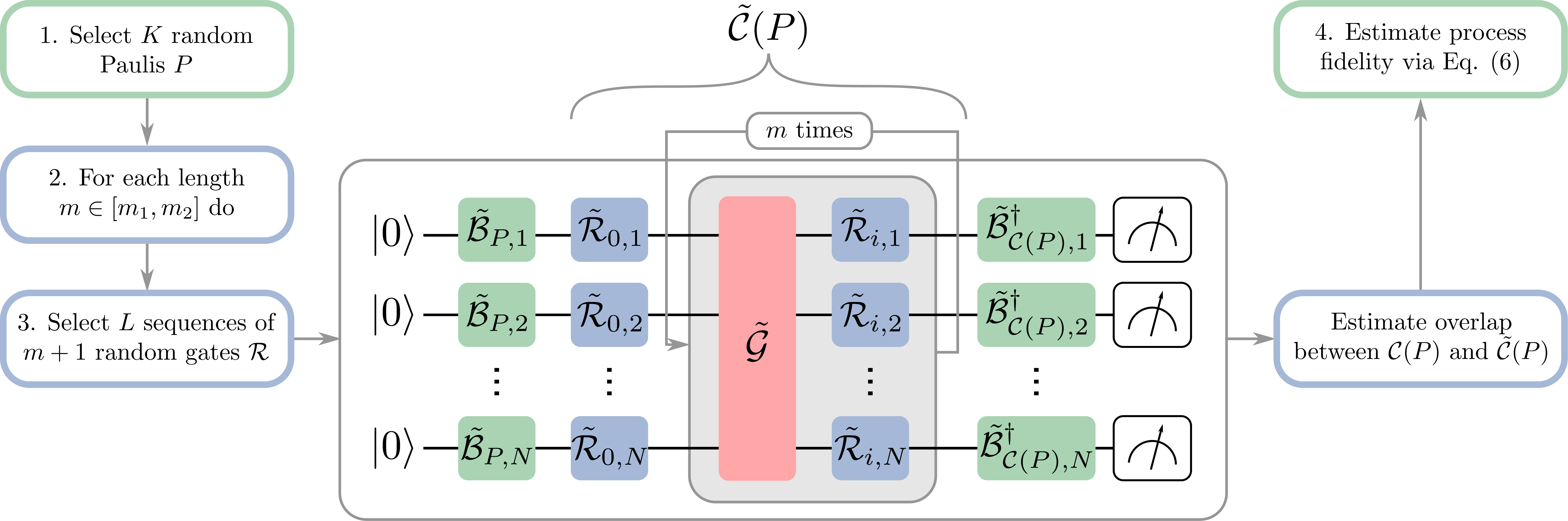}
\caption{Schematic circuit implementation of the experimental cycle benchmarking (CB) protocol. The protocol can be subdivided into three parts, depicted by the different colors.
The green gates describe basis changing operations, which are defined in the Supplementary Information. 
The red gates $\noisy{G}$ are the noisy implementations of some gate of interest.
The blue gates are random Pauli cycles that are introduced to create an effective Pauli channel per application of the gate of interest, where $\noisy{R}_{i,j}$ denotes the $j^\mathrm{th}$ tensor factor of the $i^\mathrm{th}$ gate. Creating an effective Pauli channel per application enables errors to be systematically amplified under $m$-fold iterations for more precise and SPAM-free estimation of the errors in the interleaved red gates \noisy{G}.
The blue and the red gates together form the random circuit $\noisy{C}(P)$.
The sequence of local operations before the first and last rounds of random Pauli cycles are identified as conceptually distinct but were compiled into the initial and final round of local gates in the experiment. The experimental parameters $K, m$, and $L$ of this work are given in the Supplementary Information.} 
\label{fig:CBcircuit}
\end{figure*}

The full cycle benchmarking protocol for characterizing the errors occurring under a fixed cycle of Clifford gates $\ideal{G}$ composed with a random Pauli cycle $\ideal{R}$ is as follows, illustrated in \cref{fig:CBcircuit}, where we explain the motivation for each step further below:
\begin{enumerate}
    \item Select a set of $N$-qubit Pauli matrices $\sf{P}$ with $K=|\sf{P}|$ elements.
    \item Select two lengths $m_1$ and $m_2$ such that the multiple application of $\ideal{G}$ composes to the identity $\ideal{G}^{m_1}=\ideal{G}^{m_2}=\ideal{I}$.
    \item Perform the following sequence for each Pauli matrix $P\in\sf{P}$, length $m\in (m_1, m_2)$, and $l\in(1,\ldots,L)$, where $L$ describes the number of random sequences per Pauli.
        \item[3a.] 
        Select $m+1$ random $N$-qubit Pauli cycles $\ideal{R}_0, \ideal{R}_1,\ldots, \ideal{R}_m$, and define the randomized circuit 
        \begin{align}\label{eq:circuit}
        \ideal{C}(P) = \ideal{R}_{m} \ideal{G} \ideal{R}_{m-1} \ideal{G}\ldots \ideal{R}_1\ideal{G}\ideal{R}_0
        \end{align}
        as illustrated in \cref{fig:CBcircuit}.
        \item[3b.] Calculate the expected outcome of the sequence $\ideal{C}(P)$ assuming ideal gate implementations. 
        \item[3c.] [Main experiment] Implement $\ideal{C}(P)$ and estimate the overlap 
        \begin{align}
            f_{P,m,l} = \tr [\ideal{C}(P) \, {\noisy{C}(\rho)}]
        \end{align} 
        between the expected outcome and  the noisy implementation $\noisy{C}(\rho)$ for some initial state $\rho$ that is a $+1$-eigenstate of $P$.
        State preparation and measurement are realized by applying the operations $\noisy{B}_P$ and $\noisy{B}^\dagger_{\ideal{C}(P)}$ that are described in the Supplementary Information. 
        \item 
        Estimate the composite process fidelity via
        \begin{align}\label{eq:fidelityEstimate}
    F_{\mathrm{RC}}(\noisy{G}, \ideal{G}) = \sum_{P \in \sf{P}}\frac{1}{|\sf{P}|}\left(\frac{\sum_{l=1}^L f_{P,m_{2},l}}{\sum_{l=1}^L f_{P,m_{1},l}}\right)^{\frac{1}{m_{2}-m_{1}}} \, .
\end{align}
\end{enumerate}

Step 1 ensures that the action of the $N$-qubit process is accurately estimated. 
In the Supplementary Information 
we prove that the number of Pauli matrices that need to be sampled is independent of the number of qubits, highlighting the scalability of the protocol for large quantum processors. 

Step 2 ensures that the measurement procedures for circuits in \cref{eq:circuit} with two different values of $m$ are the same.
Having the same measurement procedures for the two values of $m$ is crucial to decouple the SPAM errors from the decay in the process fidelity via the ratio in \cref{eq:fidelityEstimate}.
In our experiment, we always choose $m_{1}=4$ and $m_{2}$ to be an integer multiple of 4 as, for the considered gates, applying the operation four times subsequently yields the identity process $\ideal{G}^4=\ideal{I}$.

In step 3a, we choose random Pauli cycles to engineer an effective Pauli noise process across the $L$ randomizations.
This enables us to extract a process fidelity from the decay of $\sum_{l=1}^L f_{P,m,l}/L$ with the sequence length $m$.
This protocol is a special case of a more general protocol that can be used to efficiently characterize non-Clifford gates~\cite{patentApp} by selecting random gates and correction operators using randomized compiling~\cite{Wallman2016a} instead of Pauli frame randomization.

In step 3b, for any Clifford cycle \ideal{G}, Pauli matrix $P$, and Pauli cycles $\ideal{R}_0$, \ldots, $\ideal{R}_m$ the expected outcome of the ideal implementation $\ideal{C}(P)$ is a Pauli matrix that can be efficiently calculated.
Note that only the sign of $\ideal{C}(P)$ depends on the random Pauli cycles.
This sign is accounted for when estimating the expectation value with the procedure outlined in the Supplementary Information. 
Incorporating the sign engineers a measurement of the expectation value of $\ideal{C}(P)$ that is robust to SPAM errors, as otherwise the expectation values result from a multi-exponential decay~\cite{Carignan-Dugas2015,Helsen2018}.

In step 3c, we experimentally prepare an eigenstate of a Pauli matrix $P$, apply a circuit $\noisy{C}$ with interleaved random Pauli cycles, and measure the expectation value of $\ideal{C}(P)$.
The explicit procedures we use for preparing the eigenstate and measuring the expectation value are described in the Supplementary Information. 
As discussed in the Supplementary Information, 
the number of measurements required to estimate the expectation value to a fixed additive precision is independent of the number of qubits.

As we prove in the Supplementary Information, 
the expected value of $F_{\mathrm{RC}}(\noisy{G}, \ideal{G})$ in~\cref{eq:fidelityEstimate} for two values of $m_{1}$ and $m_{2}$ as in step 2 is equal to the composite process fidelity $F_{\mathrm{RC}}(\noisy{G}, \ideal{G})$ in \cref{eq:compositeFidelity} up to $\mathcal{O}([1-F_{\mathrm{RC}}(\noisy{G},\ideal{G})]^2)$, and always provides a lower bound.

We demonstrate the practicality of CB for multi-qubit systems by using it to experimentally estimate the process fidelity of cycles acting globally on quantum registers containing 2, 4, 6, 8, and 10 qubits.
The specific cycles we consider consist of simultaneous local Pauli gates and multi-qubit entangling M{\o}lmer-S{\o}rensen (MS) gates~\cite{sorensen1999quantum, sorensen2000entanglement} combined with simultaneous local Pauli gates.
We confine $^{40}\mathrm{Ca}^{+}$ ions in a linear Paul-trap and encode a single qubit in the electronic states of each atomic ion. The encoding utilizes the $\ket{0}=4S_{1/2}(m_j=-1/2)$ ground-state and the $\ket{1}=3D_{5/2}(m_j=-1/2)$ metastable excited state. Our quantum computing toolbox comprises independent arbitrary single qubit operations and fully entangling $N$-qubit MS gates (see Supplementary Information). 
An experimental run consists of: (i) Doppler-cooling; (ii) sideband-cooling of two the motional modes with lowest frequencies; (iii) optical pumping to the initial state $\ket{0}^{\otimes N}$; (iv) coherent manipulation; and (v) readout of the ions. Each sequence is repeated 100 times to gather statistics (for experimental details see Supplementary Information 
and Ref.~\cite{schindler2013quantum}).

Under Markovian noise, the estimate of the process fidelity from \cref{eq:fidelityEstimate} is independent of the sequence lengths $m_1$ and $m_2$ used to estimate it (see Supplementary Information).  
We tested whether our experimental apparatus satisfied this assumption by performing measurements at three values of $m$ (4, 8, and 12) on a register containing 6 qubits and comparing the results obtained from pairs of sequence lengths against each other.
The data is tabulated in the Supplementary Information, 
where the variation of the estimated fidelities is within $0.1$\,\%, which is smaller than the corresponding uncertainties of $0.4$\,\%. This suggests that the errors are Markovian and the estimated process fidelity is independent of the chosen sequence lengths for our system and henceforth we only use two sequence lengths to estimate the process fidelity.

\begin{figure}[ht]
    \centering
    \includegraphics[width=0.5\textwidth]{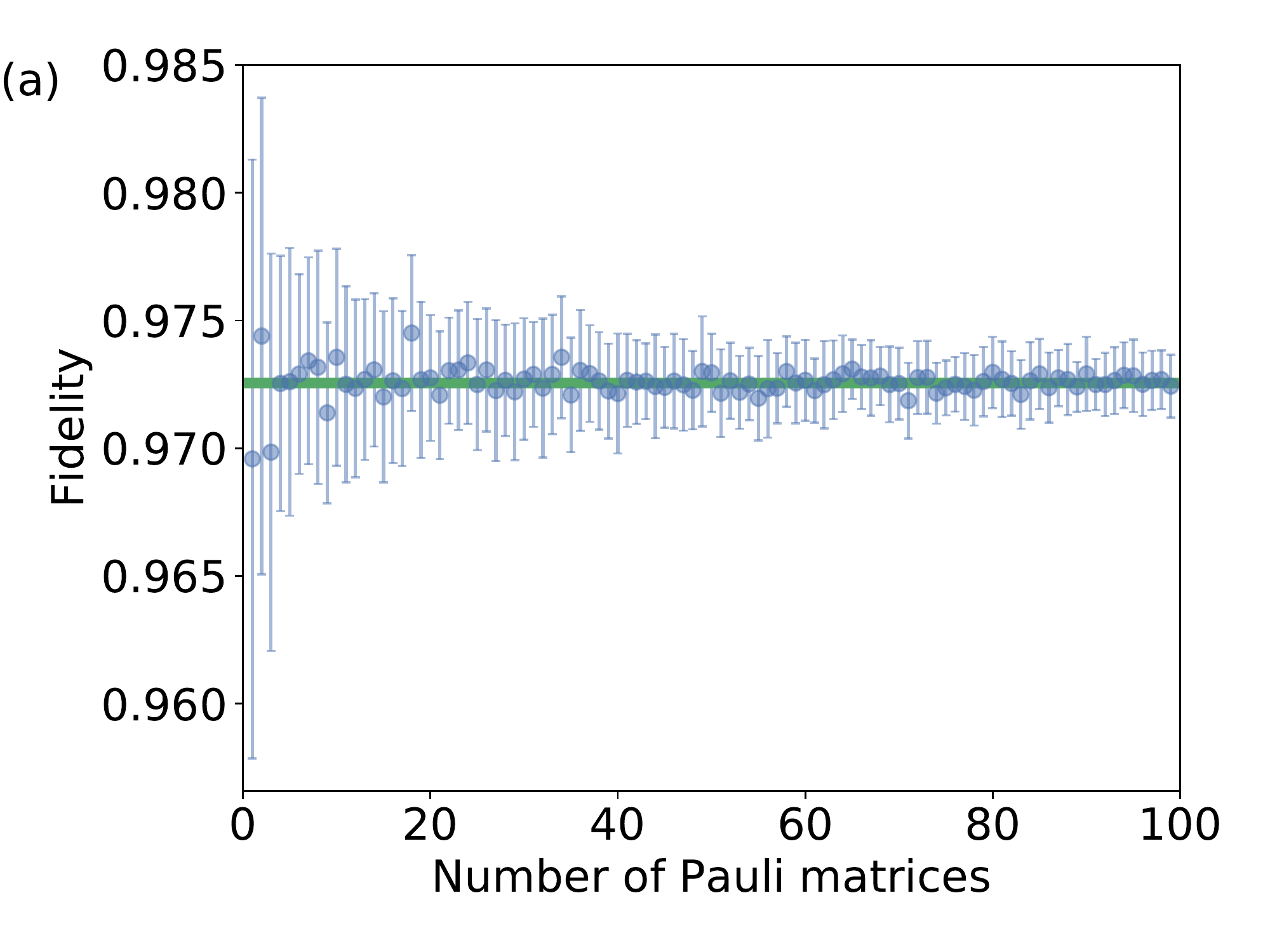}
    \qquad
    \includegraphics[width=0.5\textwidth]{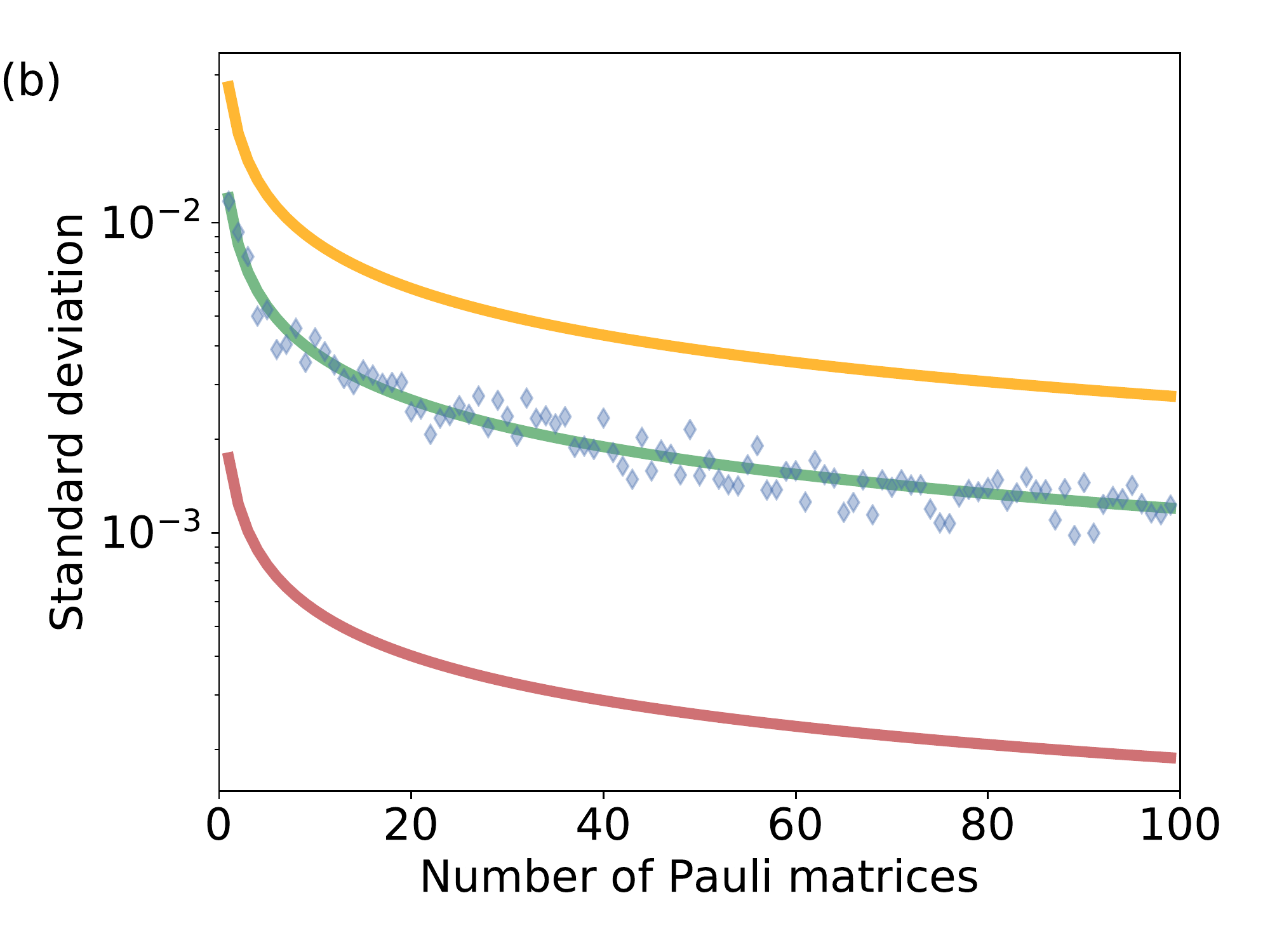}
    \caption{Experimental evidence demonstrating rapid convergence under finite sample size with favorable constant factors. (a) Mean fidelity estimates from 30 randomly sampled subsets of Pauli matrices as a function of the size of the subset. The error bars illustrate the standard deviation of the 30 samples, that is, the standard error of the mean. The green line describes the mean fidelity $\mathcal{F}=97.25(8)$\,\% calculated from the complete data set. (b) The standard deviation of the fidelity from plot (a) against $K$ including an upper bound in orange (see Supplementary Information), 
    a fit of the standard deviation data in green and a fit of the calculated projection noise in red.}
    \label{fig:subsampling}
\end{figure}

The CB protocol is practical to implement on large processors because the fidelity can be accurately estimated using a number of Pauli matrices that is \emph{independent} of the number of qubits (see Supplementary Information). 
To illustrate the rapid convergence under finite sample size, we performed CB of local Pauli operations on a 4 qubit register by exhaustively estimating all  $4^4-1=255$ possible decay rates. 
We estimate the average fidelities via \cref{eq:fidelityEstimate} for multiple subsets $\sf{P}$ of the set of all Pauli matrices.
For each $K=1,\ldots,100$, we evaluate the fidelity for 30 randomly chosen subsets $\sf{P}$ containing $|{\sf{P}}|=K$ Pauli matrices.
The mean and standard deviation of the estimated fidelities as functions of the subset size are shown in~\cref{fig:subsampling}.
The observed standard error of the mean $\sigma=0.0135(3)/\sqrt{K}$ is larger than the lower bound given by quantum projection noise,  $\sigma_\mathrm{proj}=0.00151(2)/\sqrt{K}$, but smaller than the upper bound $\sigma_\mathrm{bound}=0.0252(8)/\sqrt{K}$ on the contribution from sampling a finite number of Pauli matrices (see Supplementary Information). 
The data demonstrate that we can estimate the process fidelity $F$ to an uncertainty smaller than $(1-F)/\sqrt{K}$ using only $K \approx 20$ Pauli matrices with other experimental parameters held fixed (the parameters are listed in the Supplementary Information). 

We performed CB on local operations and with an interleaved MS gate on registers containing 2, 4, 6,
8, and 10 qubits. The process fidelity as a function of the number
of qubits in the register is shown in~\cref{fig:cb_results}
and~\cref{tab:cb_results}. While it is expected that the fidelity over the full register decreases with increasing register size,  an important question is whether the effective error rate per qubit increases, or significant cross-talk effects appear, with increasing numbers of qubits.  

\begin{figure}[ht]
    \centering
    \includegraphics[width=0.5\textwidth]{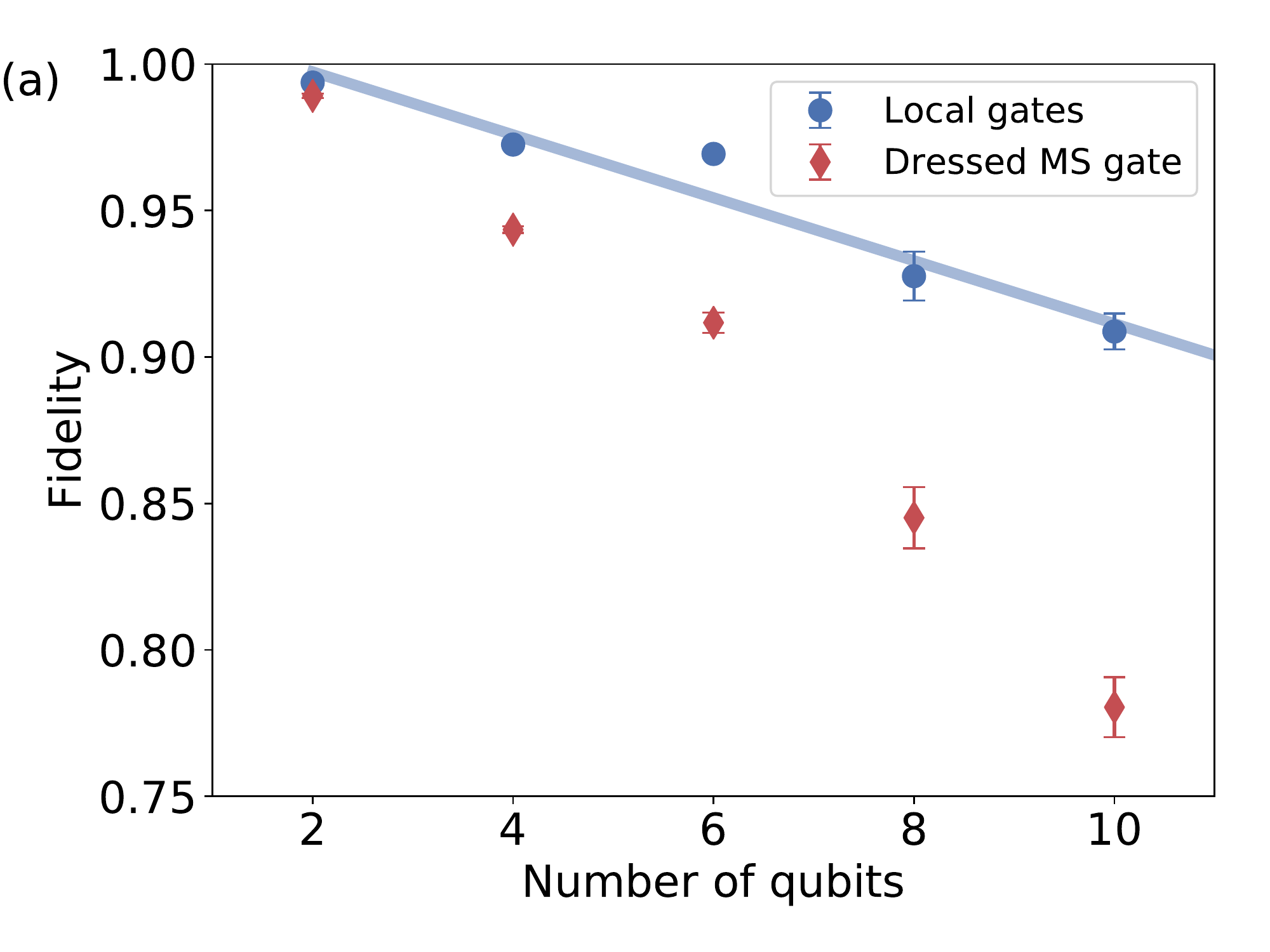}
    \qquad
    \includegraphics[width=0.5\textwidth]{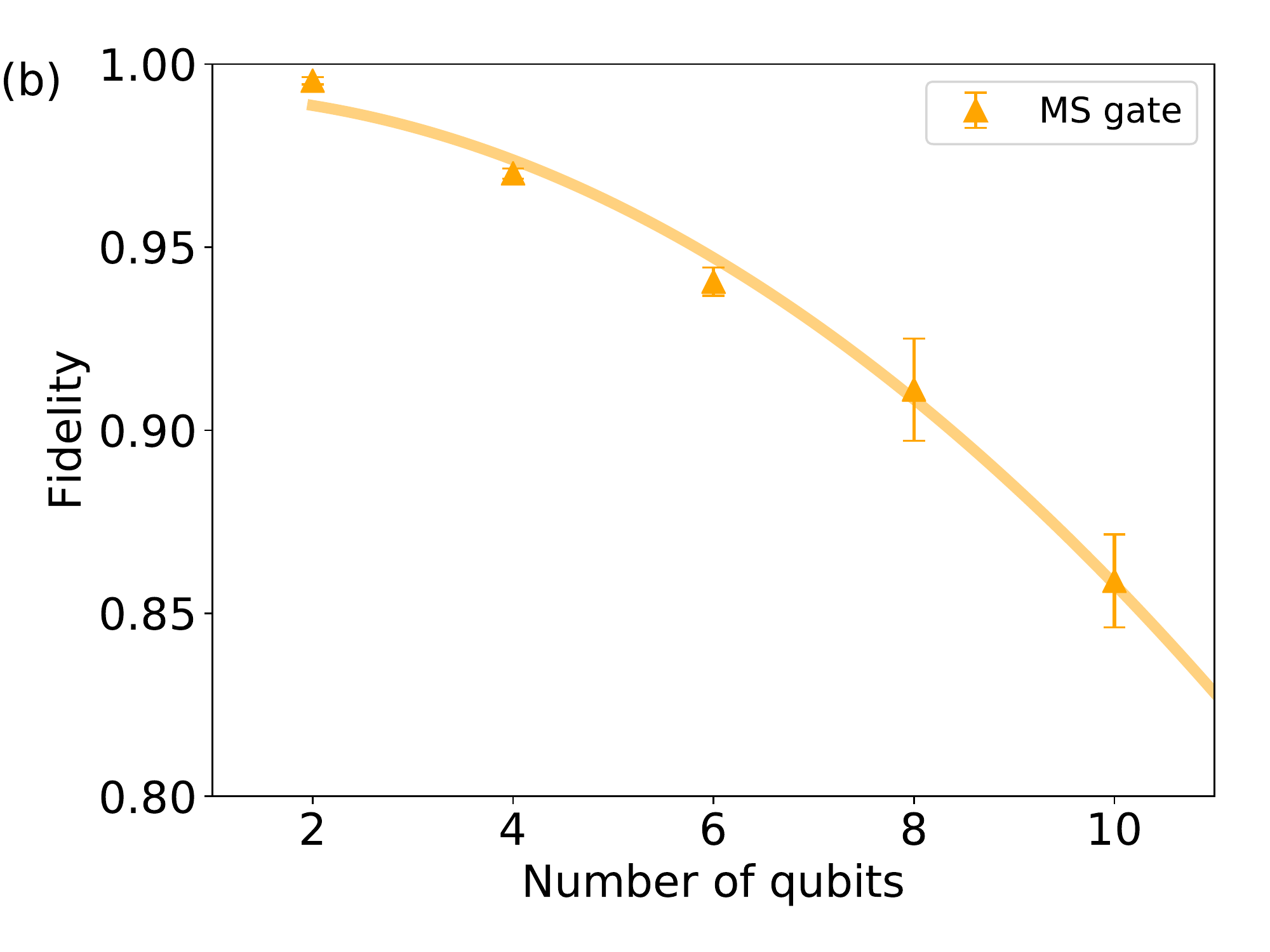}
    \caption{Experimental estimates of how rapidly error rates increase as the processor size increases. (a) Process fidelities obtained under CB for local gates (blue circles) and for sequences containing dressed MS gates (red diamonds), that is, MS gates composed with a random Pauli cycle, plotted against the number of qubits in the register. The local operations are consistent with independent errors fitted according to~\cref{eq:local_decay}. (b) Estimate of the process fidelity of an MS gate obtained by taking the ratio of dressed MS and local process fidelities. The data is fitted to~\cref{eq:ms_decay} and is consistent with a constant error per two-qubit coupling.}
    \label{fig:cb_results}
\end{figure}

We observe that the fidelity for local CB (blue circles in~\cref{fig:cb_results} (a)) decays linearly  with register size $N$, as
\begin{align}\label{eq:local_decay}
F=1-\epsilon_P N,
\end{align}
with $\epsilon_P=0.011(2)$.
The linear decay of the fidelity indicates that our single-qubit Pauli operations do not show increasing error rates per qubit or a significant onset of cross-talk errors as the register size increases. 
Each single-qubit Pauli operation requires $n_{\rm{S}}$ native gates, where on average $\langle n_{\rm{S}}\rangle=1.27$, independent of the system size.
Therefore the effective process fidelity of a native single-qubit gate is $1 - \epsilon_P / \langle n_{\rm{S}}\rangle =0.992(1)$. 

The CB measurements with interleaved MS gates give the process fidelity of the MS gate composed with a round of local randomizing gates as in \cref{eq:compositeFidelity} (a dressed MS gate, see red diamonds in~\cref{fig:cb_results} (a)).
This determines the error rate when a circuit is implemented by randomized compiling~\cite{Wallman2016a}.
The process fidelity of the interleaved gate can be estimated by the ratio of the dressed MS and local fidelities as in interleaved randomized benchmarking~\cite{Magesan2012}.
The resulting estimates are plotted in \cref{fig:cb_results} (b).
We note that these estimates may have a large systematic error that is on the same order as the overall error rate~\cite{Carignan-Dugas}.
This systematic uncertainty primarily arises due to coherent over- and under-rotations with similar rotation axes. 
The MS gate performs rotations around the non-local axes $\sigma^{(i)}_x \otimes \sigma^{(j)}_x$, which are substantially different from the single-qubit rotation axes.
Therefore it is unlikely that any coherent errors on the MS gate accumulate with the errors on the single-qubit rotations, and so we neglect this systematic error.
We conjecture that the process fidelity of the MS gate should decay quadratically due to an error in each of the $\binom{N}{2}$ couplings between pairs of qubits introduced by the MS gate.
If we assume an average error rate $\epsilon_2$ per two-qubit coupling, we can describe the MS gate fidelity as
\begin{align}\label{eq:ms_decay}
F_\mathrm{MS}=1-\epsilon_2\frac{N^2-N}{2} \, .
\end{align}
Fitting this model to the results in \cref{fig:cb_results}~(b) gives an estimated error per two-qubit coupling of $\epsilon_2=0.0030(2)$.
However, we cannot harness these two-qubit couplings individually in the experiment and thus they cannot be compared to individually available gates.

\begin{table}[ht]
\centering
\caption{Process fidelities estimated via CB (\%)}
\label{tab:cb_results}
\begin{tabular}{c|c|c|c}
Qubits & Local gates & Dressed MS gate & MS gate \\
\hline
2 & 99.37(7) & 98.92(8) & 99.6(1) \\
4 & 97.25(8) & 94.3(1) & 97.0(2) \\
6 & 96.9(2) & 91.2(3) & 94.1(4) \\
8 & 92.8(8) & 85(1) & 91(2) \\
10 & 90.9(6) & 78(1) & 86(2) \\
\end{tabular}
\end{table}

In summary, we have developed cycle benchmarking and demonstrated its practicality by implementing it on quantum registers containing
$N=2$, 4, 6, 8 and 10 qubits.
In comparison, a single random Clifford gate for 8 and 10 qubits would require more than 50 MS gates and so randomized benchmarking for 8 and 10 qubits would require a large number of measurements to achieve a useful statistical precision.
CB is practical in regimes where randomized benchmarking is impractical because it uses local randomizing gates. 
A similar approach was independently considered in \cite{Helsen2018, Xue2018} to characterize a two-qubit Clifford gate. 
However, the approach implemented here and proposed previously in Ref.~\cite{patentApp} can be applied in a scalable manner to processors with arbitrary numbers of qubits.

The total experimental time and post-processing resources required for our implementation were approximately independent of the number of qubits (see Supplementary Information), 
after accounting for the additional tests performed on specific numbers of qubits.
This is achieved because, as we prove in the Supplementary Information, 
the number $K$ of Pauli matrices that need to be sampled to estimate the fidelity is independent of the number of qubits and the fidelity.
In addition we demonstrated experimentally that the estimate of the fidelity and its error converges rapidly under finite sample size (\cref{fig:subsampling}), and that the estimated fidelities are approximately independent of the sequence lengths used. 

Cycle benchmarking can be readily implemented on general quantum computing architectures to estimate the fidelity of multi-qubit processes.
The fidelity corresponds to the effective error rate under randomized compiling~\cite{Wallman2018}.
The protocol also provides insight into how noise scales within a fixed architecture.
In our ion trap, the fidelity of local gates across the whole register decreased linearly with $N$, demonstrating that our native single-qubit gates have an average fidelity of $99.2(1)$\,\%
and do not deteriorate with the register size. Thus we have demonstrated a scalable method to validate a major
requirement for fault-tolerant quantum computation. 
In addition, we performed interleaved CB protocols to estimate the performance of the
multi-qubit entangling MS gate. From the ratio between the dressed MS and the local CB fidelities we infer entangling gate fidelities
ranging from $99.6(1)$\,\% to $86(2)$\,\% for 2 to 10 qubits.

\bibliography{rb}

\section*{Acknowledgements}
We gratefully acknowledge support by the Austrian Science Fund (FWF), through the SFB Fo-QuS (FWF Project No. F4002-N16), as well as the Institut f\"ur  Quanteninformation GmbH. This research was funded by the Office of the Director of National Intelligence  (ODNI), Intelligence Advanced Research Projects Activity (IARPA), through the Army Research Office grant W911NF-16-1-0070. All statements of fact, opinions or conclusions contained herein are those of the authors and should not be construed as  representing  the  official  views  or  policies  of IARPA,  the  ODNI,  or  the  U.S.  Government. We also  acknowledge  support  by  U.S.  A.R.O.through  grant  W911NF-14-1-0103
This research was undertaken thanks in part to funding from
TQT, CIFAR, the Government of Ontario, and the Government of Canada through CFREF, NSERC and Industry Canada.

\section*{Author contributions}
A.E., J.J.W., P.S., T.M., J.E. and R.B. wrote the manuscript and provided revisions. J.J.W., T.M. and P.S. developed the research based on discussions with J.E. and R.B.. J.J.W. and J.E. developed the theory. A.E., E.M. and P.S. performed the experiments. A.E., E.M., P.S., L.P., M.M. and R.S. contributed to the experimental setup. A.E. and J.J.W. analyzed the data. All authors contributed to discussions of the results and the manuscript.

\newcommand{\itmt}{ in the main text }
\newcommand{\eqcircuit}{\cref{eq:circuit} }
\newcommand{\eqpaulifidelity}{\cref{eq:PauliFidelity} }
\newcommand{\eqfidelityestimate}{\cref{eq:fidelityEstimate} }
\newcommand{\eqcompositefidelity}{\cref{eq:compositeFidelity} }
\newcommand{\eqprocessfidelity}{\cref{eq:processFidelity} }
\newcommand{\figcbcircuit}{\cref{fig:CBcircuit}}
\newcommand{\figsubsampling}{\cref{fig:subsampling} }

\newpage
\begin{widetext}
\Large\bfseries\centering
   Supplementary Information: Characterizing large-scale quantum computers via cycle benchmarking
\normalsize
\end{widetext}
  
\section{Theoretical methods}
In this section, we specify the state preparation and measurement (SPAM) procedures, obtain expressions for the expected values of steps in the protocol over the set of all Pauli matrices $\sf{P}^N = \{I,X,Y,Z\}^{\otimes N}$, and analyze the uncertainties in experimental estimates of those expected values.
We conclude by giving a simple expression for the ideal MS gate that facilitates the calculation of $\ideal{C}(P)$.

For this appendix only, we abuse notation slightly by implicitly defining the channel $\ideal{P}(A) = PAP^\dagger$ for any Pauli matrix $P$, so that we can use expressions such as $\sum_P \ideal{P}$.

\subsection{State preparation and measurement procedures}\label{supp:spam}

In our experiment, we can only directly perform noisy preparations and measurements in the $N$-qubit computational basis $\{|z\rangle:z\in\mathbb{Z}_2^N\}$.
We now specify the basis changes and coarse graining we use to perform other preparations and measurements.
For an $N$-qubit matrix $Q$ (e.g., $P$, $\ideal{C}(P)$ from the main text), let $\ideal{B}_Q$ rotate the computational basis to an eigenbasis of $Q$ such that
\begin{align}\label{eq:defineS}
    \sum_{z\in\mathbb{Z}_2^N} 2^{-N} \tr\left[\ideal{B}_Q(|z\rangle\!\langle z|)Q\right]\ideal{B}_Q(|z\rangle\!\langle z|) = Q.
\end{align}
For the processes we investigated, $\ideal{C}(P)$ is always an $N$-qubit Pauli matrix.
Therefore, we only need to prepare eigenstates of Pauli matrices $P$ and measure the expectation value of Pauli matrices $\ideal{C}(P)$.
Consequently, our SPAM procedures are fully specified by defining $\ideal{B}_Q$ for arbitrary Pauli matrices $Q$.
We choose to construct the $\ideal{B}_Q$ out of local Clifford operators to maximize the SPAM coefficients (which results in a smaller statistical uncertainty).
Specifically, let $P|_j$ denote the $j$th tensor factor of a matrix, $\ideal{A}_I = \ideal{A}_Z = \ideal{I}$ and 
\begin{align*}
    \ideal{A}_X(Z) = X, \quad \ideal{A}_X(X) = Y \\
    \ideal{A}_Y(Z) = Y, \quad \ideal{A}_Y(Y) = X.
\end{align*}
Then we choose the basis-changing gate for an $N$-qubit Pauli matrix $Q$ to be
\begin{align}
    \ideal{B}_Q = \bigotimes_{j=1}^N \ideal{A}_{Q|_j}.
\end{align}
Note that the basis changing procedure is independent of the sign of $Q$.

We now specify the coarse-graining procedure we use to measure the expectation value of observables.
Suppose a system is in a state $\rho$ and let $\rm{Pr}(z|Q)$ be the probability of observing the computational basis outcome $z$ after applying the process $\ideal{B}_Q^\dagger$.
One measures the expectation value of $Q$ [e.g., $Q = \ideal{C}(P)$] by applying $\ideal{B}_{Q}^\dagger$, measuring in the computational basis, and averaging the probabilities of the outcomes weighted by the coefficients $\tr\left[\ideal{B}_Q(|z\rangle\!\langle z|)Q\right]$, where the weights are computed from the ideal quantities.
From \cref{eq:defineS} and by the linearity of the trace,
\begin{align}\label{eq:expectationValue}
\tr[Q\rho]
&= \sum_{z\in\mathbb{Z}_2^N} 2^{-N} \tr\left[\ideal{B}_Q(|z\rangle\!\langle z|)Q\right]\rm{Pr}(z|Q).
\end{align}
Note that as we average the relative frequencies over all outcomes, the number of measurements required to estimate the expectation value of $Q$ to a fixed additive precision is independent of the number of qubits $N$ by a standard application of, e.g., Hoeffding's inequality~\cite{Hoeffding1963}.

The above estimation procedure will include several sources of SPAM error per qubit, including errors in qubit initialization, measuring qubits in the computational basis, and in the local processes used to change the basis.
Consequently, a protocol has to be robust to SPAM errors to provide a practical characterization of a multi-qubit gate.

\subsection{Modelling the decay as a function of the sequence length}\label{app:pauliDecays}

We now determine the expected value of $\sum_{l=1}^L f_{P,m,l}/L$ for fixed values of $P$ and $m$ under gate-independent Markovian noise on the random Pauli gates.
As in randomized compiling~\cite{Wallman2016a}, the noise on the gate of interest can be an arbitrary Markovian process.
The assumption of gate-independent noise on the random Pauli gates is weaker than the corresponding assumption in randomized benchmarking, namely, that the noise over the whole $N$-qubit Clifford group is independent of the target.
This assumption can be relaxed using the analysis of Ref.~\cite{Wallman2016a} at the cost of more cumbersome notation.

\begin{thm}\label{thm:decay}
Let \ideal{G} be a Clifford cycle and \noisy{G} be an implementation of \ideal{G} with Markovian noise.
Suppose there exists a process \ideal{A} such that $\noisy{R} = \ideal{A}\ideal{R}$ for any Pauli process \ideal{R}. 
Then for a fixed Pauli matrix $P$ and positive integer $m$, the expected value of $f_{P,m,l}$ from step 3c of the protocol over all random Pauli processes $\ideal{R}_0$, \ldots, $\ideal{R}_m$ is
\begin{align*}
\langle f_{P,m,l} \rangle = \beta \prod_{j=0}^{m-1} F_{\ideal{G}^j(P)}(\ideal{E}, \ideal{I}),
\end{align*}
where $\ideal{E} = \ideal{G}^\dagger \noisy{G}\ideal{A}$ and $\beta$ is a scalar that depends only on $P$ and $\ideal{G}^m(P)$.
Moreover, $\beta = 1$ in the absence of SPAM errors.
\end{thm}

\begin{proof}
Substituting $\noisy{R}_i = \ideal{A}\ideal{R}_i$ into the noisy version of \eqcircuit (i.e., overset each operator with a $\sim$), the average superoperator applied over all sequences for a fixed choice of random sequences is
\begin{align}
    \noisy{C} = \ideal{A}\ideal{R}_m\noisy{G}\ldots \ideal{A}\ideal{R}_1\noisy{G} \ideal{A}\ideal{R}_0.
\end{align}
Inserting $\ideal{G}\ideal{G}^\dagger$ between the ideal Pauli processes $\ideal{R}_i$ and the adjacent \noisy{G} gives
\begin{align}
    \noisy{C} = \ideal{A} \ideal{R}_m\ideal{G}\ideal{E}\ldots \ideal{R}_1\ideal{G}\ideal{E}\ideal{R}_0
\end{align}
where $\ideal{E} = \ideal{G}^\dagger \noisy{G}\ideal{A}$.
We can now do a standard relabelling of the randomizing gates to obtain a twirl by setting $\ideal{T}_0 = \ideal{R}_0$ and recursively defining
\begin{align}
    \ideal{R}_i = \ideal{T}_i \ideal{G}\ideal{T}_{i-1}^\dagger \ideal{G}^\dagger
\end{align}
for $i>0$.
With this relabelling,
\begin{align}
    \noisy{C} = \ideal{A}\ideal{T}_m \ideal{G}\ideal{T}_{m-1}^\dagger\ideal{E}\ideal{T}_{m-1}\ldots \ideal{T}_1^\dagger\ideal{E}\ideal{T}_1\ideal{G}\ideal{T}_0^\dagger\ideal{E}\ideal{T}_0.
\end{align}
The $\ideal{T}_i$ are all Pauli processes because $\ideal{G}\ideal{P}\ideal{G}^\dagger$ is a Pauli process for any Pauli process $\ideal{P}$ and any Clifford process $\ideal{G}$.
Moreover, the $\ideal{T}_i$ are uniformly random because the Pauli processes are sampled uniformly at random and form a group.
Therefore averaging independently over all $\ideal{T}_0$, \ldots, $\ideal{T}_{m-1}$ for a fixed choice of $\ideal{T}_m$ results in the effective superoperator
\begin{align}\label{eq:averageCircuit}
    \ideal{A}\ideal{T}_m(\ideal{G}\noisy{E})^m,
\end{align}
where
\begin{align}\label{eq:pauliTwirl}
    \noisy{E} = 4^{-N} \sum_{P\in\sf{P}^N} \ideal{P}^\dagger \ideal{E} \ideal{P}.
\end{align}
Now note that \noisy{E} is invariant under conjugation by Pauli operators and so $\noisy{E}(Q) \propto Q$ for all $Q\in\sf{P}^N$~\cite{Holevo2005}.
As the Pauli matrices form a trace-orthogonal basis for the set of matrices, 
\begin{align}\label{eq:Covariant}
    \noisy{E}(Q) 
    &= 2^{-N}\tr\left[Q^\dagger \noisy{E}(Q)\right] Q \notag\\
    &= 4^{-N} \sum_{P\in\sf{P}^N} 2^{-N} \tr\left[Q\ideal{P}^\dagger \ideal{E} \ideal{P}(Q)\right]Q \notag\\
    &= 4^{-N} \sum_{P\in\sf{P}^N} 2^{-N} \tr\left[\ideal{P}(Q)\ideal{E} \ideal{P}(Q)\right]Q \notag\\
    &= 4^{-N} \sum_{P\in\sf{P}^N} 2^{-N} \tr\left[Q\ideal{E}(Q)\right]Q \notag\\
    &= F_Q(\ideal{E}, \ideal{I})Q,
\end{align}
for any $Q\in \sf{P}^N$, where we have used the fact that $\ideal{P}(Q) = PQP^\dagger = \pm Q$ for any Pauli matrices $P,Q$ and \eqpaulifidelity.

For any two Pauli matrices $P,Q\in\sf{P}^N$, let
\begin{align}
\eta(Q,P) = \begin{cases}
1 & \mbox{ if } QP = PQ \\
-1 & \mbox{otherwise.}\\
\end{cases}
\end{align}
Then, from \cref{eq:averageCircuit} with $P' = \ideal{G}^m(P)$ for convenience, the expected outcome of the ideal circuit is $\ideal{C} = \eta(T_m, P;) P'$.
Now note that under measurement errors and noisy changes of basis [i.e., errors in the $\mathrm{Pr}(z|Q)$] and folding the residual \ideal{A} into the measurement, \cref{eq:expectationValue} gives the expectation value of some operator $\tilde{P}'$ (which is not uniquely defined).
Since only the weights in \cref{eq:expectationValue} depend on the sign of $P'$ and are calculated from the ideal expressions, the noisy measurement for $-P'$ gives the expectation value of $-\tilde{P}'$ by linearity.

Let $\rho$ be the prepared state after applying a noisy change of basis.
Then the expectation value of $f_{P,m,l}$ in step 3c over all sequences is
\begin{align}\label{eq:av}
    \langle f_{P,m,l} \rangle &= 4^{-N} \sum_{T_m\in\sf{P}^N} \eta(T_m,P')\tr \left[\ideal{T}_m^\dagger(\tilde{P}') (\ideal{G}\noisy{E})^m(\rho)\right] \notag\\
    &= \alpha_P\tr\left[P' (\ideal{G}\noisy{E})^m(\rho)\right]
\end{align}
by \cref{lem:project} below, where $\alpha_P = 2^{-N}\tr[P \tilde{P}']$ is 1 in the absence of errors.

Expanding $\rho = \sum_{Q\in\sf{P}^N} \rho_Q Q$ and noting that \ideal{G} is a Clifford cycle, \cref{eq:av} reduces to
\begin{align}
 \langle f_{P,m,l} \rangle = \sum_{Q\in\sf{P}^N} \alpha_P \rho_Q 
    \tr \left[P'\ideal{G}^m(Q)\right]\prod_{j=0}^{m-1} F_{\ideal{G}^j(Q)}(\ideal{E}, \ideal{I}).
\end{align}
As the Pauli matrices are trace-orthogonal and $P' = \ideal{G}^m(P)$, $\tr \left[\ideal{G}^m(Q)P'\right] = 2^N \delta_{Q,P}$.
Therefore
\begin{align}\label{eq:fidelityDecay}
 \langle f_{P,m,l} \rangle = 2^N \alpha_P \rho_P 
    \prod_{j=0}^{m-1} F_{\ideal{G}^j(P)}(\ideal{E}, \ideal{I}),
\end{align}
where $\rho_P = 2^{-N}$ in the absence of SPAM errors, so that $\beta=2^N \alpha_P \rho_P = 1$ in the absence of SPAM errors.
\end{proof}

In the above proof, we make use of the following lemma proven and applied to randomized benchmarking in Ref.~\cite{Helsen2018}.

\begin{lem}\label{lem:project}
For any matrix $M$ and any Pauli matrix $P$,
\begin{align*}
    4^{-N} \sum_{Q\in\sf{P}^N} \eta(Q,P)\ideal{Q}(M) = 2^{-N}\tr\left[P M\right] P.
\end{align*}
\end{lem}

\begin{proof}
As the Pauli matrices form an orthogonal basis for the space of matrices, we can write 
\begin{align}
    M = \sum_{R\in \sf{P}^{\otimes N}} m_R R,
\end{align}
where $m_R = 2^{-N} \tr(R M)$.
As $\ideal{Q}(R) = \eta(Q, R) R$ for any Pauli matrix $R$,
\begin{align}
4^{-N} \sum_{Q\in\sf{P}^N} \eta(Q, P)\ideal{Q}(M) 
= \sum_{R\in \sf{P}^{\otimes N}} m_R (\eta_P \cdot \eta_R) R
\end{align}
by linearity, where
\begin{align}
    \eta_P\cdot \eta_R = 4^{-N} \sum_{Q\in\sf{P}^N} \eta(Q,R)\eta(P,R).
\end{align}
As $\eta(Q,P)$ is a real 1-dimensional representation of the Pauli group for any fixed Pauli matrix $P$ and $\eta(Q,P)$ and $\eta(Q,R)$ are inequivalent as representations for $P\neq R$, 
\begin{align}
    4^{-N} \sum_Q \eta(P^{(m)}, Q)\eta(P, Q)
    &= \delta(P,R)
\end{align}
by Schur's orthogonality relations.
\end{proof}

\subsection{Estimating the process fidelity}\label{app:precision}

We now prove that the expectation value of \eqfidelityestimate provides an accurate, yet conservative, estimate of the process fidelity in \eqcompositefidelity under the same assumptions as in \eqpaulifidelity.

\begin{thm}\label{thm:precision}
Let
\begin{align*}
    \hat{F} &= 4^{-N} \sum_{P\in\sf{P}^N} \left(\frac{\langle f_{P,m_{2},l} \rangle }{\langle f_{P,m_{1},l} \rangle }\right)^{\frac{1}{m_{2}-m_{1}}}
\end{align*}
be the expected outcome of the cycle benchmarking protocol over all randomizations.
Let \ideal{G} be a Clifford cycle and \noisy{G} be an implementation of \ideal{G} with Markovian noise.
Suppose there exists a process \ideal{A} such that $\noisy{R} = \ideal{A}\ideal{R}$ for any Pauli process \ideal{R}. 
Then $\hat{F} \leq F_{\mathrm{RC}}(\noisy{G},\ideal{G})$ and
\begin{align*}
    \hat{F} - F_{\mathrm{RC}}(\noisy{G},\ideal{G}) = \mathcal{O}\left([1-F_{\mathrm{RC}}(\noisy{G},\ideal{G})]^2\right).
\end{align*}
\end{thm}

\begin{proof}
First, recall that the process fidelity is linear and for any unitary process \ideal{U},
\begin{align*}
    F(\noisy{G},\ideal{U}) = F(\ideal{U}^\dagger \noisy{G},\ideal{I}).
\end{align*}
Therefore from \eqcompositefidelity,
\begin{align*}
    F_{\mathrm{RC}}(\noisy{G}, \ideal{G})
    &= 4^{-N} \sum_{R\in\sf{P}^N} F(\noisy{G}\noisy{R}, \ideal{G}\ideal{R}) \\
    &= 4^{-N} \sum_{R\in\sf{P}^N} F(\ideal{R}\ideal{G}^\dagger\noisy{G}\ideal{A}\ideal{R}, \ideal{I})\\
    &= F(\noisy{E}, \ideal{I}).    
\end{align*}
Moreover, $F(\ideal{E},\ideal{I}) = F(\noisy{E},\ideal{I})$ by \eqprocessfidelity and \cref{eq:Covariant},
and so we will prove statements for $F(\ideal{E},\ideal{I})$.

Now fix a Pauli matrix $P$ and note that if $m_1$ and $m_2= m_1+\delta\!m$ are chosen so that $P'=\ideal{G}^{m_2}(P)=\ideal{G}^{m_1}(P)$ (guaranteed by step 2 of the protocol), then
\begin{align}
    \left(\frac{\langle f_{P,m_2,l}\rangle}{\langle f_{P,m_1,l}\rangle}\right)^{1/\delta\!m} = \prod_{j=0}^{\delta\!m - 1} F_{\ideal{G}^j(P')}(\ideal{E}, \ideal{I})^{1/\delta\!m}
\end{align}
by \cref{thm:decay}, as the scalar is the same for $m_1$ and $m_2$.
That is, the terms being averaged over in \eqfidelityestimate are themselves geometric means of $F_Q(\noisy{E}, \noisy{I})$ for different Pauli matrices $Q$ obtained by applying $\ideal{G}$ to the sampled $P$.
Formally, let $w(Q|P',\delta\!m)$ be the relative frequency of $Q$ in the list $(\ideal{G}^j(P'):j=0,\ldots, \delta\!m-1)$.
Then
\begin{align}\label{eq:geometricMean}
    \left(\frac{\langle f_{P,m_2,l}\rangle}{\langle f_{P,m_1,l}\rangle}\right)^{1/\delta\!m} = \prod_{Q\in\sf{P}^N} F_Q(\ideal{E}, \ideal{I})^{\omega(Q|\ideal{G}^{m_1}(P),\delta\!m)}
\end{align}
By the inequality of the weighted arithmetic and geometric means,
\begin{align}\label{eq:ratio}
    \left(\frac{\langle f_{P,m_2,l}\rangle}{\langle f_{P,m_1,l}\rangle}\right)^{1/\delta\!m} \leq \sum_{Q\in\sf{P}^N} w(Q|P,\delta\!m)F_Q(\ideal{E}, \ideal{I}).
\end{align}
As $\ideal{G}$ is a Clifford matrix, $\sum_{P\in\sf{P}^N} \omega(Q|P,\delta\!m) = 1$ for all Pauli matrices $Q$.
Therefore summing \cref{eq:ratio} over all input Pauli matrices $P$ gives $\hat{F} \leq F(\ideal{E},\ideal{I})$.
To prove the approximate statement, let $r_Q = 1 - F_Q(\ideal{E}, \ideal{I})$.
Expanding \cref{eq:geometricMean} to second order in the $r_Q$ gives
\begin{align*}
    \left(\frac{\langle f_{P,m_2,l}\rangle}{\langle f_{P,m_1,l}\rangle}\right)^{1/\delta\!m} = 1 - \sum_{Q\in\sf{P}^N} \omega(Q|P,\delta\!m) r_Q + \mathcal{O}(r_Q^2).
\end{align*}
The approximate claim then holds as $\mathcal{O}(r_Q^2) = \mathcal{O}([1-F(\ideal{E},\ideal{I})]^2)$ by \cref{lem:pauli_infidelities} below.
\end{proof}

\begin{lem}\label{lem:pauli_infidelities}
For any completely positive and trace-preserving map $\ideal{E}$ and any Pauli matrix $P$, 
\begin{align*}
0 \leq 1 - F_P(\ideal{E},\ideal{I}) \leq 2 - 2F(\ideal{E},\ideal{I}).
\end{align*}
\end{lem}

\begin{proof}
Note that \cref{eq:Covariant} holds for any completely positive and trace preserving map \ideal{E} with \noisy{E} as defined in \cref{eq:pauliTwirl}.
In particular, $F_P(\noisy{E},\ideal{I}) = F_P(\ideal{E},\ideal{I})$ for all $P\in\sf{P}^N$ and so $F(\noisy{E},\ideal{I}) = F(\ideal{E},\ideal{I})$ by \eqprocessfidelity.
As $\noisy{E}$ is covariant under Pauli channels, there exists a probability distribution $p(Q)$ over the set of Pauli matrices such that~\cite{Holevo2005}.
\begin{align}\label{eq:Holevodecomp}
\noisy{E}(A) = \sum_{Q} p(Q) QAQ^\dagger.
\end{align}
For any Kraus operator decomposition, the process fidelity can be written as~\cite{Nielsen2002}
\begin{align}\label{eq:fidIprob}
F(\noisy{E}, \ideal{I}) &= \sum_Q p(Q) |\tr Q|^2/4^N = p(I).
\end{align}
Substituting \cref{eq:Holevodecomp} into \eqpaulifidelity and using $[P,I]=0$, $p(Q)\geq 0$, and \cref{eq:fidIprob} gives
\begin{align}\label{eq:pauli_bound}
F_P(\noisy{E}, \ideal{I}) &= \sum_{Q:[Q,P]=0} 2p(Q) - 1 \notag\\
&\geq 2p(I) - 1 = 2F(\noisy{E},\ideal{I}) - 1.
\end{align}
The lower bound follows as the $F_P(\noisy{E},\ideal{I})$ are eigenvalues of \noisy{E} and hence are in the unit disc~\cite{Evans1978}.
\end{proof}

\subsection{Finite sampling effects}\label{app:sampling}

We now consider the effect of finite samples.
All the ``approximately normal'' statements in this section can be replaced by rigorous statements using the results of \cite{Harper2019}, Hoeffding's inequality~\cite{Hoeffding1963} and the union bound, at the expense of additional notation and less favorable (but pessimistic) constants.

First, note that sampling a finite number of random sequences (i.e., finite $L$) and estimating each expectation value with a finite number of measurements will produce an estimate of $\langle f_{P,m,l}\rangle$ with an error $\epsilon_{P,m}$ that is approximately normally distributed with standard deviation $\sigma_{P,m}$.
Using a series expansion of the ratio 
\begin{align*}
    \hat{F}_P := \left(\frac{\langle f_{P,m_2,l} \rangle}{\langle f_{P,m_1,l} \rangle}\right)^{1/\delta\!m},
\end{align*}
the error in each term in the sum will be approximately $(\epsilon_{P,m_2} - \epsilon_{P,m_1})/\delta\!m$ and so will be approximately normally distributed.
Moreover, if we choose $m_1$ and $m_2$ so that $\delta\!m \approx 1 - F(\ideal{E},\ideal{I})$ (where \ideal{E} is as in \cref{thm:decay}), then the error on each term in the sum will have standard deviation $\sigma_P \propto 1 - _P(\ideal{E},\ideal{I})$.
The values of $m$ in \cref{tab:parameters} satisfy this condition.

We now consider the effect of sampling a finite number $K$ of Pauli matrices $P$ with replacement under the same assumptions as in \cref{thm:decay}.
Sampling $K$ Pauli matrices $P$ uniformly at random with replacement and averaging the estimates $\hat{F}_P$ gives an estimate $\hat{F}$ whose expected variance over the Pauli matrices is
\begin{align}
    \mathbb{V}^2(\hat{F})&=
     \frac{\mathbb{V}^2(\hat{F}_P)}{K} + \sum_P\frac{ \mathbb{V}^2(\sigma_P^2)}{K}.
\end{align}
The first term satisfies
\begin{align}\label{eq:finitePauliBound}
     \mathbb{V}^2(\hat{F}_P) \leq [1-F(\ideal{E},\ideal{I})]^2
\end{align}
since for any Pauli matrix $P$,
\begin{align}
    |F(\ideal{E},\ideal{I}) - F_P(\ideal{E},\ideal{I})| 
    &\leq \max_{Q\in\sf{P}^N} |F(\ideal{E},\ideal{I}) - F_P(\ideal{E},\ideal{I})| \notag\\
    &\leq 1 - F(\ideal{E}, \ideal{I})
\end{align}
by \cref{lem:pauli_infidelities}.
Note that the variance is independent of the number of qubits.
Furthermore, if the $\delta\!m$ are chosen to be proportional to $1/(1-\hat{F})$, then the variance of $\hat{F}$ is proportional to $(1-\hat{F})^2$, so that we can efficiently estimate $1-\hat{F}$ to \textit{multiplicative} precision.

It can be seen in \figsubsampling that the standard deviation decreases with the square-root of the sampled subspaces $K$, with a least squares fit giving $\sigma=0.0135(3)/\sqrt{K}$. The observed standard deviation is larger than the lower bound given by quantum projection noise  $\sigma_\mathrm{proj}=0.00151(2)/\sqrt{K}$ but smaller than the upper bound $\sigma_\mathrm{bound}=0.0252(8)/\sqrt{K}$ on the contribution from sampling a finite number of Pauli matrices.
This suggests that the other source of statistical uncertainty, namely, a finite number of randomizations $L$ and measurements per sequence, is sufficiently small to allow us to accurately estimate the process fidelity.

\subsection{Correction operators for the MS gate}\label{sec:MS_corrections}

We performed cycle benchmarking for the identity and MS gates.
The $\rm{MS}$ gate satisfies $\rm{MS}^4 = I$, so that we can restrict $m$ to be an integral multiple of $4$.
Indeed, $\rm{MS}^2 \propto X\tn{N}$ so that we could restrict $m$ to be even numbers by keeping track of the sign (which would depend on the Pauli matrix $P$).
To compute the expectation value of $\ideal{C}(P)$, we need to know how an arbitrary Pauli operator $Q$ propagates through the $\rm{MS}$ gate.
Using $\mathrm{MS} \propto (I - iX\tn{N})/\sqrt{2}$ for even $N$ gives
\begin{align}\label{eq:MSPauliMap}
\ideal{MS}(Q)
&= \mathrm{MS}Q\mathrm{MS}^\dagger \notag\\
&= \begin{cases}
Q & \mbox{ if } QX\tn{N} = X\tn{N} Q \\
i QX\tn{N}\mathrm{MS} & \mbox{otherwise.}
\end{cases}
\end{align}

\section{Experimental methods}
\label{sec:experimental methods}
\begin{table*}[ht]
\centering
\caption{Experimental parameters for the taken CB data for different register.}
\label{tab:parameters}
\begin{tabular}{c|c|c|c|c|c}
Qubits & Subspaces $K$ & Sequence lengths $m$ & Random sequences $L$ & Total sequences & Measurement time (h) \\
\hline
2 & 15 & 4, 40 & 10 & 600 & 2.6 \\
4 & 255 & 4, 20 & 10 & 10200 & 15.7 \\
6 & 43 & 4, 8, 12 & 10 & 2580 & 3.4 \\
8 & 24 & 4, 8 & 10 & 960 & 2.0 \\
10 & 21 & 4, 8 & 10 & 840 & 1.9 \\
\end{tabular}
\end{table*}
The CB experiments are defined by a sequence of $N$-qubit Clifford gates according to the experimental protocol in \figcbcircuit. Specifically, the sequences contain a series of single qubit rotations and $N$-qubit MS gates. A rotation of qubit $j$ with angle $\theta$ is defined as $R(\theta)_j=\mathrm{exp}(i\theta p_j/2)$, where $p_j\in[X,Y,Z]$ are single-qubit Pauli operations.

After defining the sequences we compile them into the actual machine language~\cite{martinez2016compiling}. In this experiment an elementary single qubit operation consist of one addressed $z$-rotation sandwiched between two collective rotations around the $x$- or $y$-axis, e.g. $X(\pi/2)_1=X(-\pi/2)_{12}Z(\pi/2)_1X(\pi/2)_{12}$ for 2 qubits.
The collective $x$- and $y$-rotations can be seen as simple basis changes on the entire register, and thus these basis changes can be shared by the individual qubit operations. By changing the temporal order of the collective $x$-, $y$-rotations and the individual $z$-rotations, the total number of collective rotations can be minimized. 

We expect the single qubit $z$-rotations to have significantly larger infidelity compared to the collective rotations for the following reasons: First, the addressed laser beam has a smaller beam size and hence has larger intensity fluctuations. Second, we perform the $z$-rotations using the AC-Stark effect, which is quadratically more sensitive to intensity fluctuations than resonant $x$-, $y$-rotations. Therefore the number of single qubit rotations $Z(\theta)_j$ needed to perform a $N$-qubit Pauli operation is expected to be the limiting factor for local operations. In general, the average number of single qubit rotations per $N$-qubit Pauli operation scales linearly with $N$. To simplify the calibration procedure we only perform $Z(\pi/2)_j$ rotations. Thus e.g. a $Z(\pi)_j$ operation is implemented using two $Z(\pi/2)_j$ operations. In \cref{fig:z_pulses} we show the dependency of the average number of $Z(\pi/2)_j$ operations on the number of qubits. On average we implement $1.27(2)\cdot N$ addressed $\pi/2$ rotations for an $N$-qubit Pauli operation.
\begin{figure}[ht]
    \centering
    \includegraphics[width=0.5\textwidth]{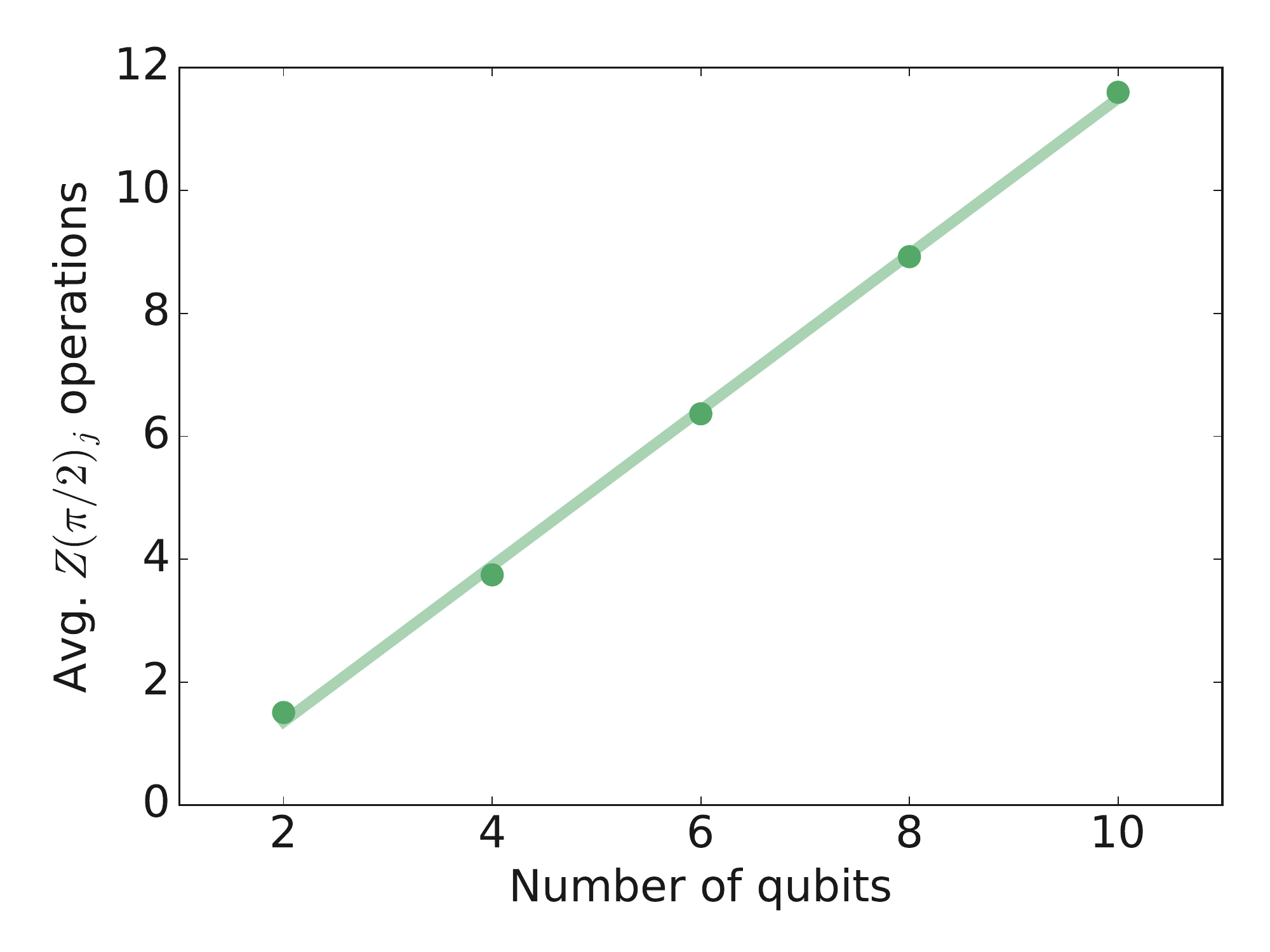}
    \caption{Average number of $Z(\pi/2)_j$ operations needed to implement a $N$-qubit Pauli gate.}
\label{fig:z_pulses}
\end{figure}

In \cref{tab:parameters} we give an overview of the experimental parameters that we used to estimate the local CB and the dressed MS fidelities.

\subsection{Testing the dependence of the estimator on the sequence length}
If the noise in the system is Markovian, we expect the estimated process fidelity to be independent of the sequence lengths $m_1$ and $m_2$ to within $\mc{O}([1-F_{\mathrm{RC}}(\noisy{G},\ideal{G})]^2)$ (see \cref{thm:precision}).
We test this by performing measurements at 3 different sequence lengths for 6 qubits, as described in \cref{tab:parameters}. 
We validate that the estimated process fidelity is independent of $m_1$ and $m_2$ by comparing the results of three different length pairs 4-8, 4-12 and 8-12. 
As can be seen in \cref{tab:model_test}, the measured fidelities agree to within half a standard deviation, which supports the validity of the assumtions for our experimental apparatus.
\newline
\begin{table}[ht]
\centering
\caption{6-qubit process fidelities estimated via CB (\%) using different pairs sequence lengths $(m_1,m_2)$. The results illustrate that the estimated process fidelity is independent of the sequence lengths used, subject to the constraint in step 2 of the protocol.}
\label{tab:model_test}
\begin{tabular}{c|c|c}
$(m_1,m_2)$ & Local gates & Dressed MS gate \\
\hline
(4,8) & $97.0(2)$ & $91.3(5)$ \\
(4,12) & $97.0(2)$ & $91.2(4)$ \\
(8,12) & $96.9(4)$ & $91.3(8)$ \\
\end{tabular}
\end{table}

\subsection{Analyzing fidelity drift}
Slow temperature fluctuations on the timescale of minutes to days cause changes in various components of our experimental apparatus. One of the major causes for a loss in fidelity over time is the alignment of the laser beams relative to the ion position. The single ion addressing laser beam is tightly focused to a spot size of $\sim 2$\,$\mu$m and the beam position changes as the temperature varies. This change in position leads to a miscalibration of the Rabi frequency as well as an increase in intensity fluctuations. We analyze the temporal dependence of the fidelity with 4-qubit CB as depicted in \cref{fig:drift}. The 255 subspaces were measured in 3 sessions, where the experimental system was recalibrated at the beginning of each session. We approximate the drift of the fidelity to be linear in first order and thus can describe the time dependent fidelity as $F(t)=F_0-\epsilon t$. We obtain an average loss of fidelity of $\epsilon_\mathrm{L}=3.3(5)\cdot 10^{-3}$\,h$^{-1}$ for local gates and $\epsilon_\mathrm{I}=5.4(8)\cdot 10^{-3}$\,h$^{-1}$ for the dressed MS gate, see \cref{tab:drift}. This measurement suggests that we can expect a maximum loss of fidelity of 1\,\% when recalibrating the apparatus every two hours.
\begin{figure*}[hbt]
\centering
\includegraphics[width=.25\textwidth]{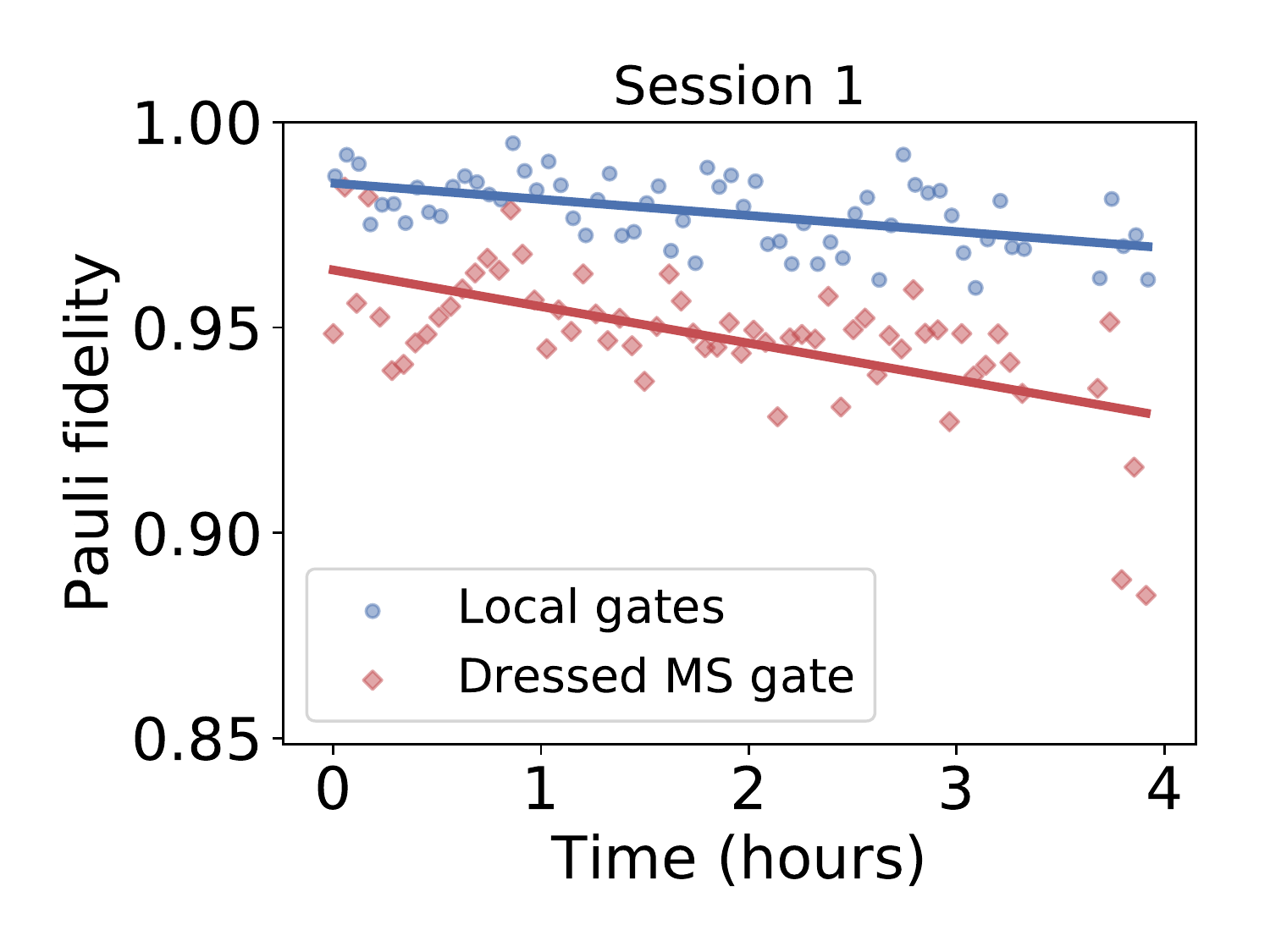}\hfill
\includegraphics[width=.5\textwidth]{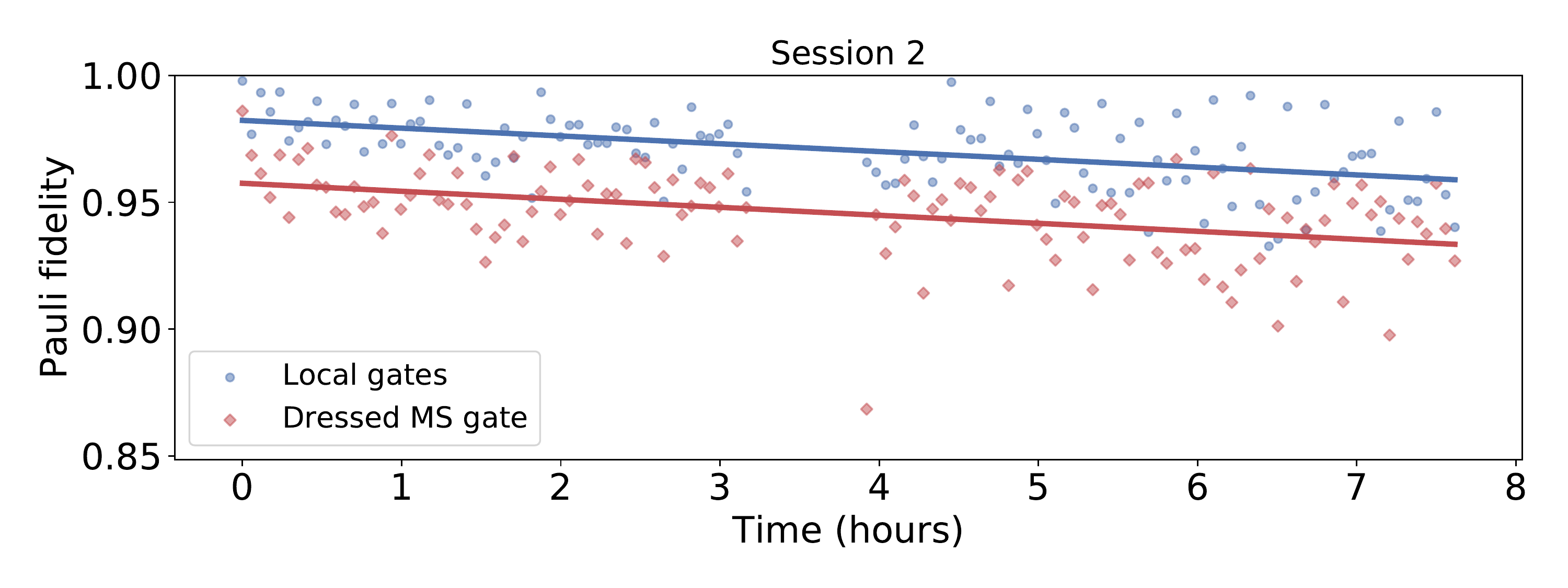}\hfill
\includegraphics[width=.25\textwidth]{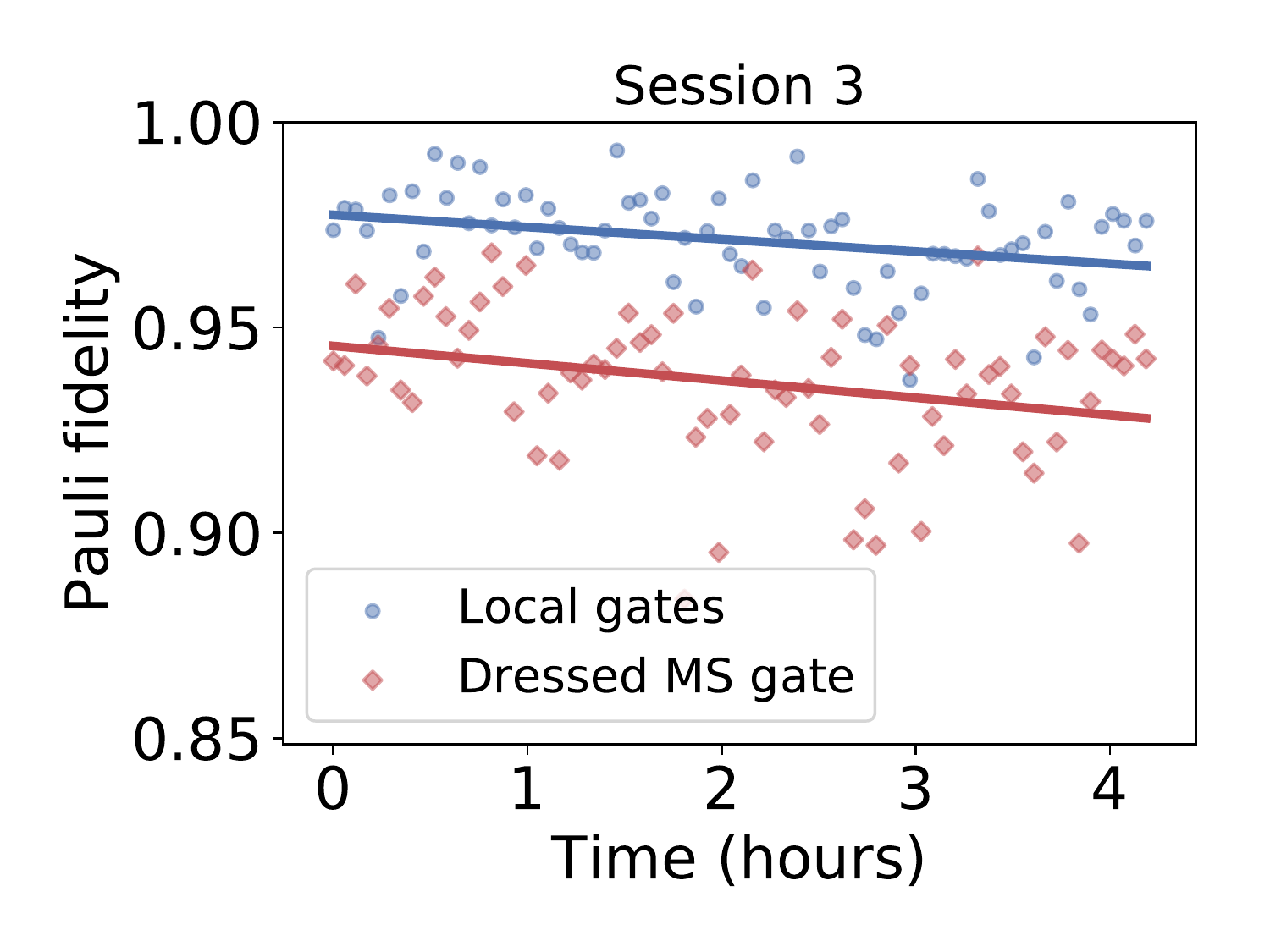}
\caption{4-qubit Pauli fidelities for local gates (blue) and the dressed MS gate (red) plotted on the time in hours. We measured all 255 subspaces in three measurement sessions, where the experiment was recalibrated at the beginning of each session.}
\label{fig:drift}
\end{figure*}
\begin{table}[ht]
\centering
\caption{4-qubit fidelity drift rates, where $\epsilon_\mathrm{L}$ and $\epsilon_\mathrm{D}$ describe the loss of fidelity per hour for local gates and the dressed MS gate. The data corresponds to the estimated linear slopes of~\cref{fig:drift}}
\label{tab:drift}
\begin{tabular}{c|c|c}
Session & $\epsilon_\mathrm{L}$ (h$^{-1}$) & $\epsilon_\mathrm{D}$ (h$^{-1}$) \\
\hline
1 & $3.9(8)\cdot 10^{-3}$ & $8.9(1.5)\cdot 10^{-3}$ \\
2 & $3.1(5)\cdot 10^{-3}$ & $3.2(6)\cdot 10^{-3}$ \\
3 & $3.0(1.1)\cdot 10^{-3}$ & $4.2(1.7)\cdot 10^{-3}$\\
\hline
Average & $3.3(5)\cdot 10^{-3}$ & $5.4(8)\cdot 10^{-3}$ \\
\end{tabular}
\end{table}

\end{document}